\documentclass[letterpaper, 12 pt, onecolumn, conference]{ieeeconf}
\usepackage{amsfonts}
\usepackage{amssymb,amsmath,epsfig,graphicx,amsmath,epstopdf,color}
\usepackage[T1]{fontenc}

\IEEEoverridecommandlockouts
\newtheorem{theorem}{Theorem}

\newtheorem{remark}{Remark}

\newtheorem{property}{Property}

\begin{document}

\title{Desired Model Compensation based Position Constrained Control of Robotic Manipulators }
\author{ Samet Gul, Erkan Zergeroglu, Enver Tatlicioglu and Mesih Veysi Kilinc \thanks{
\newline S. Gul is with the Department of Computer Engineering, Gebze Technical University, 41400, Gebze, Kocaeli, Turkey (Email: sametgul@gtu.edu.tr),
\newline E. Zergeroglu is with the Department of Computer Engineering, Gebze Technical University, 41400, Gebze, Kocaeli, Turkey (Email: e.zerger@gtu.edu.tr),
\newline E. Tatlicioglu is with the Department of Electrical \& Electronics Engineering, Izmir Institute Pof Technology, Izmir, 35430 Turkey  E-mail: enver@iyte.edu.tr),
\newline M. V. Kilinc is with the Institute of Information Technologies, Gebze Technical University, 41400, Gebze, Kocaeli, Turkey (Email: m.kilinc@gtu.edu.tr).
}}

\maketitle

\begin{abstract}
This work presents the design and the corresponding stability analysis of desired model based, joint position constrained, robot controller. Specifically, provided that the initial joint position tracking error signal starts below some predefined value, the proposed controller ensures that the joint  tracking error signal remains inside the region (defined by predefined upper--bound) and approaches to zero asymptotically. 


\end{abstract}

\section{Introduction}

One of the drawback of controller designs for MIMO systems based on Lyapunov type analysis techniques is the lack of direct knowledge of the transient performance of the system states. As the outcome of the stability analysis via Lyapunov based arguments is usually stated with respect to increasing time. Specifically, when the overall stability result obtained through a Lyapunov based analysis is asymptotic stability of the system states, we can conclude that the states of the system remain bounded (i.e. in $\mathcal{L}_{\infty}$ ) and will eventually converge to the desired states, but this does not give any information on how the states behave during the transient. However, on real world systems, transient behavior is as important as the type of the stability obtained, as this also frames the steady state behavior and if direct manipulation of transient behavior is not possible then at least a reasonable bound should be ensured preferably \textit{a priori}. 

A possible solution to this problem relies on barrier Lyapunov function (BLF) approach based designs. Although, applying constraints was considered in optimization field for quite some time, its application to nonlinear control field is relatively new and dates back to early 2000s \cite{jiang05}, \cite{Tee09}. Some line of the past studies applied BLFs to deal with constraints for systems in the Brunovsky form \cite{jiang05}, in strict feedback form \cite{Tee09}, in strict feedback form with time varying output constraints \cite{Tee11}, and in pure feedback form \cite{Liu16}. An asymmetric BLF was proposed for systems in pure feedback form under time varying full state constraints in \cite{Wang17}. In \cite{Afflitto18}, bounding both the trajectory tracking error and the parameter estimation error vector within user defined constraint sets have been considered.

One line of research have focused on applying BLF type control designs to mechatronic systems. In \cite{kabzinski17}, a systematic motion controller based on BLF was designed for servo systems. \cite{Doulgeri09} and \cite{Doulgeri10icra} used prescribed performance criteria for regulation control of robot manipulators which are extended to tracking control in \cite{Doulgeri12RAS}, \cite{Doulgeri13ROB}, \cite{Doulgeri16RAL}. Hackl and Kennel, in \cite{Hackl12}, designed a position controller with prescribed performance criteria for robot manipulators when dynamic model is partially known. In \cite{Doulgeri10iros}, a task space regulator guaranteeing prescribed performance was proposed while \cite{Doulgeri16} used BLF based approach for joint limit avoidance sub--task of redundant robot manipulator control problem. In \cite{Doulgeri12}, prescribed performance methods were utilized for referential control of human--like movements of redundant arms. \cite{Doulgeri12AUT}, \cite{Doulgeri10} considered force/position control of robot manipulators with prescribed performance while guaranteeing contact.

In this work, tracking control of robot manipulators in joint space is aimed. In addition to the joint position tracking objective, ensuring \textit{a priori} limits for the entries of the joint tracking error is aimed as the secondary control objective. The control problem is further complicated due to the presence of parametric uncertainties in the mathematical model of the robot manipulator. To ensure tracking control objective, a nonlinear proportional derivative feedback component is designed. To deal with the parametric uncertainties, desired compensation based adaptive controller component is proposed as part of the control input. Different from the similar designs in the literature, in a novel approach, the control gain of the tracking error is proposed to be error--dependent and two fundamentally different control gain matrices are designed. To ensure limiting the entries of the tracking error in addition to guaranteeing asymptotic convergence to the origin, two BLFs are introduced. Numerical simulation results are shown to be commensurate with the analysis.

\section{Robot Model and Properties}

In this section, the mathematical model of the robot along with some of the model properties that will be made use of during the analysis will be presented. The mathematical model of an $n$ degree of freedom (dof) revolute joint robot manipulator is presented as \cite{lewis}
\begin{equation}
M\left(q\right) \ddot{q} + C\left(q ,\dot{q}\right) \dot{q} + G\left(q\right) + F_{d}\dot{q} = \tau \label{model}
\end{equation}
in which $q\left(t\right)$, $\dot{q}\left(t\right)$, $\ddot{q}\left(t\right) \in \mathbb{R}^{n}$ are the vectors for joint positions, velocities, and accelerations, respectively, $M\left(q\right) \in \mathbb{R}^{n\times n}$ is the inertia matrix, $C\left(q,\dot{q}\right) \in \mathbb{R}^{n\times n}$ is the centripetal Coriolis matrix, $G\left(q\right) \in \mathbb{R}^{n}$ represents the gravitational effects, $F_d \in \mathbb{R}^{n\times n}$ is a positive definite diagonal matrix denoting the constant viscous frictional effects, and $\tau\left(t\right) \in \mathbb{R}^{n}$ is the control input torque. As commonly utilized in the robotics literature, the dynamic model terms in \eqref{model} satisfy the standard properties detailed below.

\begin{property}\label{P1}
The inertia matrix is positive definite and symmetric, and also satisfies the given inequalities \cite{lewis}
\begin{equation}
m_1 I_{n} \leq M \left(q\right) \leq m_2 I_{n}  \label{prop1}
\end{equation}
where $m_1$, $m_2$ are known, positive, bounding constants with $I_n \in \mathbb{R}^{n\times n}$ being the standard identity matrix.
\end{property}

\begin{property}\label{P2}
The inertia and centripetal Coriolis matrices satisfy the given skew--symmetry relationship \cite{lewis}
\begin{equation}
\xi^T \left(\dot{M}-2C\right) \xi = 0 \text{ } \forall \xi \in \mathbb{R}^{n}. \label{prop2}
\end{equation}
\end{property}

\begin{property}\label{P3}
The centripetal Coriolis matrix satisfies the given switching expression \cite{lewis}
\begin{equation}
C\left(\xi ,\nu \right) \eta = C\left(\xi ,\eta\right) \nu \text{ } \forall \xi, \nu, \eta \in \mathbb{R}^{n}. \label{prop3}
\end{equation}
\end{property}

\begin{property}\label{P4}
Following bounding expressions can be obtained for the dynamic model terms in \eqref{model} \cite{lewis}
\begin{eqnarray}
\Vert M\left( \xi \right) - M\left( \nu \right) \Vert_{i\infty } &\leq &\zeta_{m1} \Vert \xi - \nu \Vert \label{prop4a} \\
\Vert C\left( \xi , \nu \right) \Vert_{i\infty } &\leq &\zeta_{c1} \Vert \nu \Vert \label{prop4b} \\
\Vert C\left( \xi , \nu \right) - C\left( \eta , \nu \right) \Vert_{i\infty } &\leq &\zeta_{c2} \Vert \xi -\eta \Vert \label{prop4c} \\
\Vert G\left( \xi \right) - G\left( \nu \right) \Vert &\leq &\zeta_{g} \Vert \xi - \nu \Vert \label{prop4d}
\end{eqnarray}
$\forall$ $\xi $, $\nu$, $\eta \in \mathbb{R}^{n}$ where $\zeta_{m1}$, $\zeta_{c1}$, $\zeta_{c2}$, $\zeta_{g}\in \mathbb{R}$ are positive
bounding constants and subscript $i\infty$ denoting induced infinity norm of a matrix.
\end{property}

\begin{property}\label{P5}
The mathematical model of the robot dynamics given in \eqref{model} can be reconfigured as
\begin{equation}
Y\left(q,\dot{q},\ddot{q}\right)\theta = M\left( q\right) \ddot{q} + C\left( q,\dot{q}\right)\dot{q} + G\left( q\right) + F_{d}\dot{q} \label{prop5}
\end{equation}
in which $Y\left(q,\dot{q},\ddot{q}\right)\in \mathbb{R}^{n\times p}$ is the regression matrix that is a function of the joint position, velocity and acceleration vectors, and $\theta \in \mathbb{R}^{p}$ contains constant robot model parameters. The regression matrix formulation of \eqref{prop5} is also written in terms of desired trajectory and its time derivatives in the following manner
\begin{equation}
Y_d\left(q_{d},\dot{q}_{d},\ddot{q}_{d}\right)\theta = M\left( q_{d}\right) \ddot{q}_{d} + C\left( q_{d},\dot{q}_{d}\right) \dot{q}_{d} + G\left( q_{d}\right) + F_{d}\dot{q}_{d} \label{prop5d}
\end{equation}
where the desired version of the regression matrix $Y_d\left(q_{d},\dot{q}_{d},\ddot{q}_{d}\right)\in \mathbb{R}^{n\times p}$ is a function of the desired joint position, velocity and acceleration vectors denoted respectively by $q_{d}\left( t\right)$, $\dot{q}_{d}\left( t\right)$, $\ddot{q}_{d}\left( t\right) \in \mathbb{R}^{n}$.
\end{property}

\section{Control Problem and Error System Development}

In this section, the control problem and the accompanying error system development will be presented. The primary control objective is to design the control input torque $\tau \left( t\right) $ such that the joint position vector $q\left( t\right)$ approaches to the desired joint position vector $q_{d}\left( t\right) $ as time increases (\textit{i.e.}, the tracking control objective). In addition to the joint tracking control objective, a secondary control objective is to ensure that the entries of the position tracking error, shown with $e\left( t\right) \in \mathbb{R}^{n}$, remains inside a predefined bound, denoted with $\Delta_i>0$ for each joint $i$, in the sense that\footnote{In this paper, subscript $i$ of a diagonal matrix or a column vector denote the $i$th diagonal entry of the matrix or the $i$th entry of the vector, respectively.}
\begin{equation}
\vert e_{i}\left(t\right) \vert < \Delta_i \text{ } \forall t>0 . \label{obj}
\end{equation}
Providing the stability of the closed loop system by keeping all system trajectories bounded is also essential. In the subsequent development, joint position and joint velocity measurements are considered to be available. The control problem is complicated due to parametric uncertainties in the robot dynamic model (\textit{i.e.}, $\theta$ vector in \eqref{prop5} or \eqref{prop5d} is uncertain). The desired joint position trajectory is considered to be chosen as sufficiently smooth in the sense that itself along with its first two derivatives are bounded functions of time.

In order to quantify the main control objective, the joint position tracking error $e\left( t\right) \in \mathbb{R}^{n}$ is defined as
\begin{equation}
e \triangleq q_{d}-q \label{e}
\end{equation}
and to present the subsequent design and the associated synthesis and analysis with only first time derivatives, a filtered error, shown with $r\left( t\right) \in \mathbb{R}^{n}$, is introduced
\begin{equation}
r \triangleq \dot{e} + \alpha e \label{r}
\end{equation}Türkiye bu krize, zayıf bir Merkez Bankası, zayıf bir bankacılık sistemi, zayıf bir bütçeyle girdi.
where $\alpha \in\mathbb{R}^{n\times n}$ is a constant, positive definite, diagonal control gain matrix. To obtain the open loop error system dynamics, the time derivative of $r\left( t\right)$ is premultiplied with the inertia matrix to re
ach
\begin{equation}
M\left( q\right) \dot{r} = M\left( q\right) \left(\ddot{q}_{d}+\alpha \dot{e}\right) + C\left(q,\dot{q}\right) \left(\dot{q}_{d}+\alpha e\right) - C\left( q,\dot{q}\right) r + G\left( q\right) + F_{d} \dot{q} - \tau \label{Mrdot0}
\end{equation}
where \eqref{model} was substituted into and \eqref{r} was made use of along with the time derivative of \eqref{e}. Adding and subtracting the desired robot dynamics in \eqref{prop5d} to the right hand side of \eqref{Mrdot0} deduces to
\begin{equation}
M\left( q\right) \dot{r} = - C\left( q,\dot{q}\right) r + \chi + Y_{d} \theta -\tau \label{Mrdot1}
\end{equation}
where $\chi\left(q,\dot{q}, q_d, \dot{q}_d, \ddot{q}_d \right) \in\mathbb{R}^n$ is an uncertain vector defined as
\begin{equation}
\chi \triangleq M\left( q\right) \left(\ddot{q}_{d}+\alpha \dot{e}\right) + C\left(q,\dot{q}\right) \left(\dot{q}_{d}+\alpha e\right) + G\left( q\right) + F_{d}\dot{q} - Y_{d}\theta . \label{chi}
\end{equation}
Via making use of the Properties \ref{P1}, \ref{P3} and \ref{P4}, $\chi$ can be proven to be upper bounded as
\begin{equation}
\Vert \chi \Vert \leq \rho_{1} \Vert e \Vert + \rho_{2} \Vert r \Vert \label{chiBound}
\end{equation}
with $\rho_1\left( \Vert e \Vert \right)$ and $\rho_2\left( \Vert e \Vert \right)$ being known, non--negative, non--decreasing functions. 

\section{Control Design}

In this section, the design of the control law including the adaptation mechanism that will compensate for the uncertain model parameters will be presented. Based on the subsequent stability analysis, the control input torque $\tau\left( t\right)$ is designed as
\begin{equation}
\tau = Y_{d}\hat{\theta} + K_{r} r + K_{e} e + v_{R} \label{tau}
\end{equation}
where $K_{r}\in \mathbb{R}^{n\times n}$ is a constant, positive definite, diagonal control gain matrix, $K_{e}\left( e \right)\in \mathbb{R}^{n\times n}$ is yet to be designed tracking error dependent, positive definite, diagonal control gain matrix, and $v_{R} \left(e,r\right)\in\mathbb{R}^n$ is introduced to compensate for the negative effects of $\chi$ and is designed as
\begin{equation}
v_{R} = \left( k_{n} \rho_{1}^{2} + \rho_{2}\right) r \label{vr}
\end{equation}
with $k_n\in\mathbb{R}$ being a constant, positive damping gain, and $\hat{\theta} \left(t\right)\in\mathbb{R}^p$ is the parameter estimation vector that is adaptively updated according to
\begin{equation}
\dot{\hat{\theta}} = \Gamma Y_{d}^{T} r \label{update}
\end{equation}
in which $\Gamma \in \mathbb{R}^{p\times p}$ is a constant, positive definite, diagonal adaptation gain matrix.

Substituting the designed control input torque in \eqref{tau} and \eqref{vr} into the open loop error system in \eqref{Mrdot1} deduces the below closed loop error system
\begin{equation}
M\left( q\right) \dot{r} = - C\left( q,\dot{q}\right) r  - K_{r} r - K_{e} e + \chi - \left( k_{n} \rho_{1}^{2} + \rho_{2}\right) r + Y_{d} \tilde{\theta} \label{Mrdot2}
\end{equation}
with $\tilde{\theta}\left(t\right)\in\mathbb{R}^p$ denoting the parameter estimation error defined as
\begin{equation}
\tilde{\theta} \triangleq \theta - \hat{\theta} . \label{tildetheta}
\end{equation}

Introduction of the tracking error dependent control gain matrix $K_{e}$ is the main difference of this work from similar past research in the literature in the sense that the design of $K_e$ will enable us to continue with a novel stability analysis to ensure \textit{a priori} boundedness of the entries of the tracking error with ``user imposed'' upper bounds. For this aim, two different $K_e$ designs are proposed. The first one is designed as\footnote{The notation $\text{diag} \lbrace \cdot \rbrace$ denotes a diagonal matrix with its diagonal entries being the ones in the braces.}
\begin{equation}
K_{e} = \text{diag} \left\{ \frac{k_{i}}{\Delta_{i}^{2} - e_{i}^{2}} \right\} \label{Keln}
\end{equation}
with $k_{i}$ $i\in \lbrace 1, \cdots ,n \rbrace$ being constant gains, while the second one is designed as
\begin{equation}
K_{e} = \text{diag} \left\{ 1+\tan^{2}\left( \frac{\pi}{2} \frac{e_{i}^{2}}{\Delta_{i}^{2}}\right) \right\} \label{Ketr}
\end{equation}
with $\Delta_{i}$ being previously introduced in \eqref{obj}.

\section{Stability analysis}

In this section, the stability analysis will be presented. Despite the two tracking error dependent control gain matrix designs being fundamentally different, only the initial parts of the stability analysis differ. For both designs, the stability analysis is framed by the following theorem.

\begin{theorem}
For the robot manipulator mathematical model in \eqref{model}, the controller in \eqref{tau} and \eqref{vr} along with the adaptive update law in \eqref{update} and the tracking error dependent control gain matrix design in either \eqref{Keln} or \eqref{Ketr} ensures global asymptotic convergence of the tracking error and the filtered error to the origin and guarantees that the entries of the joint tracking error remain within a predefined bound while also proving closed loop stability by ensuring boundedness of all the system trajectories provided that the damping gain $k_n$ introduced in \eqref{vr} is chosen sufficiently high. 
\end{theorem}
\begin{proof}
For the tracking error dependent diagonal controller gain matrix design in \eqref{Keln}, the analysis is initiated by defining the barrier Lyapunov function $V_{l}\left( r,e,\tilde{\theta}\right) \in\mathbb{R}$ as
\begin{equation}
V_{l} \triangleq \frac{1}{2}r^{T}M\left( q\right) r + \sum_{i=1}^{n}\frac{k_{i}}{2}\ln\left( \frac{\Delta_{i}^{2}}{\Delta_{i}^{2}-e_{i}^{2}}\right) + \frac{1}{2}\tilde{\theta}^{T}\Gamma^{-1}\tilde{\theta} \label{Vl}
\end{equation}
which is positive definite and radially unbounded provided that the initial values of the entries of the joint tracking error satisfies $\vert e_{i}\left( 0\right) \vert < \Delta_{i}$ for all $i\in \lbrace 1, \cdots , n\rbrace$. 

Türkiye bu krize, zayıf bir Merkez Bankası, zayıf bir bankacılık sistemi, zayıf bir bütçeyle girdi.
Taking the time derivative of \eqref{Vl} yields
\begin{equation}
\dot{V}_{l} = r^{T} M\left( q\right) \dot{r} + \frac{1}{2}r^{T}\dot{M}\left( q\right) r + \sum_{i=1}^{n} k_{i} \frac{e_{i} \dot{e}_{i}}{\Delta_{i}^{2}-e_{i}^{2}} + \tilde{\theta}^{T}\Gamma ^{-1}\dot{\tilde{\theta}} \label{Vldot1}
\end{equation}
in which 
\begin{equation}
\sum_{i=1}^{n} k_{i} \frac{e_{i} \dot{e}_{i}}{\Delta_{i}^{2}-e_{i}^{2}} = e^T K_e \dot{e} \label{Vldot1a}
\end{equation}
in view of the diagonal structure of $K_e$ in \eqref{Keln}. Substituting the closed loop error system in \eqref{Mrdot2} for $r$ dynamics, \eqref{Vldot1a} and \eqref{r} for $e$ dynamics, the time derivative of \eqref{tildetheta} along with \eqref{update} for $\tilde{\theta}$ dynamics into \eqref{Vldot1} deduces
\begin{eqnarray}
\dot{V}_{l} &=& r^{T} \left[- C\left( q,\dot{q}\right) r  - K_{r} r - K_{e} e + \chi - \left( k_{n} \rho_{1}^{2} + \rho_{2}\right) r + Y_{d} \tilde{\theta} \right] \nonumber \\
&& + \frac{1}{2}r^{T}\dot{M}\left( q\right) r + e^T K_e \left(-\alpha e + r \right) - \tilde{\theta}^{T} Y_d^T r . \label{Vldot2}
\end{eqnarray}
At the right hand side of \eqref{Vldot2}, making use of the skew--symmetry relationship in Property \ref{P2}, upper bounding $\chi$ with \eqref{chiBound} and then canceling common terms give
\begin{equation}
\dot{V}_{l} \leq - r^{T} K_{r} r - e^T K_e \alpha e + \left[ \rho_{1} \Vert e \Vert \Vert r \Vert - k_{n} \rho_{1}^{2} \Vert r \Vert^2 \right] \label{Vldot3}
\end{equation}
in which for the square bracketed term \cite{Kokotovic92}
\begin{equation}
\rho_{1} \Vert e \Vert \Vert r \Vert - k_{n} \rho_{1}^{2} \Vert r \Vert^2 \leq \frac{1}{4 k_n} \Vert e \Vert^2 \label{damping1}
\end{equation}
can be used to further obtain an upper bound as
\begin{equation}
\dot{V}_{l} \leq - \lambda_{\min} \lbrace K_{r} \rbrace \Vert r \Vert^2 - \left( \lambda_{\min} \lbrace K_e \alpha \rbrace - \frac{1}{4 k_n} \right) \Vert e \Vert^2 . \label{Vldot4}
\end{equation}
After defining the combined error vector $x \triangleq \left[ \begin{array}{cc} r^{T} & e^{T} \end{array} \right]^{T} \in\mathbb{R}^{2n}$ and positive constant $\beta\in\mathbb{R}$ as
\begin{equation}
\beta \triangleq \min \left\{ \lambda_{\min} \lbrace K_{r} \rbrace , \lambda_{\min} \lbrace K_e \alpha \rbrace - \frac{1}{4 k_n} \right\} \label{beta1}
\end{equation}
following upper bound can be obtained for the right hand side of \eqref{Vldot4}
\begin{equation}
\dot{V}_{l} \leq - \beta \Vert x \Vert^2 \label{Vldot5}
\end{equation}
provided that $k_n$ is chosen sufficiently high. 

From the structures of \eqref{Vl} and \eqref{Vldot5}, $V_{l}$ is proven to be bounded and thus $r\left(t\right)$, $e\left(t\right)$, $\tilde{\theta}\left(t\right) \in \mathcal{L}_{\infty}$. By utilizing the boundedness of the above terms along with the boundedness of the desired trajectory and its time derivatives, $\dot{e}\left(t\right)$, $\dot{r}\left(t\right)\in \mathcal{L}_{\infty}$ can be proven from \eqref{r} and \eqref{Mrdot2}, respectively. It can straightforwardly be shown that the remaining terms can be ensured to be bounded as well. After integrating \eqref{Vldot5} on time from initial time to infinity, $x\left(t\right)$ is proven to be square integrable and thus $r\left( t\right)$, $e\left( t\right)\in \mathcal{L}_{2}$. Since $x\left( t\right)\in \mathcal{L}_{2} \cap \mathcal{L}_{\infty}$ and $x\left( t\right)\in  \mathcal{L}_{\infty}$, from Barbalat's Lemma \cite{khalil} $x\left(t\right) \to 0$ as $t \to \infty$ is proven.

For the tracking error dependent diagonal controller gain matrix design in \eqref{Ketr}, the analysis is initiated by defining a similar barrier Lyapunov function, denoted by $V_{t}\left( r,e,\tilde{\theta}\right) \in\mathbb{R}$, as
\begin{equation}
V_{t} \triangleq \frac{1}{2}r^{T}M\left( q\right) r + \sum_{i=1}^{n}\frac{\Delta_i^2}{\pi} \tan\left( \frac{\pi}{2} \frac{e_{i}^{2}}{\Delta_i^2}\right) + \frac{1}{2}\tilde{\theta}^{T}\Gamma ^{-1} \tilde{\theta} \label{Vt}
\end{equation}
where only the second term is different than the corresponding term in \eqref{Vl}. It is noted that $V_{t}\left( r,e,\tilde{\theta}\right) $ is positive definite and radially unbounded provided that the initial values of the entries of the joint tracking error satisfy $\vert e_{i} \left( 0\right) \vert < \Delta_{i}$ for all $i\in \lbrace 1, \cdots ,n\rbrace$.

The time derivative of \eqref{Vt} is obtained as
\begin{equation}
\dot{V}_{t} = r^{T}M\left( q\right) \dot{r} + \frac{1}{2}r^{T}\dot{M}\left( q\right) r + \sum_{i=1}^{n} e_i \dot{e}_i \left( 1 + \tan^2\left( \frac{\pi}{2} \frac{e_{i}^{2}}{\Delta_i^2}\right) \right) + \tilde{\theta}^{T}\Gamma ^{-1} \dot{\tilde{\theta}} \label{Vtdot1}
\end{equation}
where, in view of the diagonal structure of $K_e$ design in \eqref{Ketr}, the third term can be reformulated as
\begin{equation}
\sum_{i=1}^{n} e_i \dot{e}_i \left( 1 + \tan^2\left( \frac{\pi}{2} \frac{e_{i}^{2}}{\Delta_i^2}\right) \right) = e^T K_e \dot{e} . \label{Vtdot1a}
\end{equation}
A closer look at the structure of \eqref{Vtdot1} in view of \eqref{Vtdot1a} reveals the fact that it is same as \eqref{Vldot1} used with \eqref{Vldot1a}, thus the rest of the analysis is same as the previous part. 
\end{proof}

\begin{remark}
Despite obtaining the same result for both choices of tracking error dependent control gain matrices, due to the differences in their designs in \eqref{Keln} and \eqref{Ketr}, the resulting $\beta$ values may be different. For the design in \eqref{Keln}, $\lambda_{\min} \lbrace K_e \alpha \rbrace$ is equal to $\min \lbrace \frac{k_{i}\alpha_{i}}{\Delta_{i}^{2} - e_{i}^{2}} \rbrace$ over $i\in \lbrace 1, \cdots , n \rbrace$, which can conservatively be obtained as $\frac{\min \lbrace k_{i}\alpha_{i}\rbrace}{\max \lbrace \Delta_{i}^{2} \rbrace}$. On the other hand, for the design in \eqref{Ketr}, $\lambda_{\min} \lbrace K_e \alpha \rbrace$ is equal to $\min \lbrace \alpha_{i} \left( 1+\tan^{2}\left( \frac{\pi}{2} \frac{e_{i}^{2}}{\Delta_{i}^{2}}\right) \right) \rbrace$ over $i\in \lbrace 1, \cdots , n \rbrace$ which, after noting that the $\tan^2 (\cdot)$ term being nonnegative, can be conservatively obtained as $\min \lbrace \alpha_{i} \rbrace$. 
\end{remark}

\section{Conclusions}

In this work, we have presented the design and the corresponding analysis of two different types of full state feedback, desired model based, joint position constrained, robot controllers using barrier Lyapunov functions. The proposed controllers ensure that the position tracking error of the robot joints remain inside a predefined value and eventually converge to zero when the initial tracking error starts inside this predefined region, despite the presence of uncertainties in the parameters of robot dynamics. Future work will concentrate on output feedback version of the proposed method and  extending this result to task space control of robotic manipulators.

\end{document}